\documentclass[10pt,letter,epsfig]{IEEEtran}
\usepackage{amsthm}
\usepackage{amsmath}
\usepackage{algorithmic}
\usepackage[ruled]{algorithm2e}
\usepackage{mathrsfs}
\usepackage{graphicx}
\usepackage{subfigure}
\usepackage{cite}
\usepackage{bm,epsfig,amsthm,url}
\usepackage{indentfirst}
\usepackage{amssymb}
\usepackage{epstopdf}
\usepackage{xcolor}
\usepackage{mathtools}

\theoremstyle{definition}

\newtheorem{proposition}{Proposition}

\begin{document}
	\title{\LARGE Multi-Active/Passive-IRS Enabled Wireless Information and Power Transfer: Active IRS Deployment and Performance Analysis}
	
	\author{ Min~Fu,~\IEEEmembership{Member,~IEEE}, Weidong~Mei,~\IEEEmembership{Member,~IEEE}, and~Rui~Zhang,~\IEEEmembership{Fellow,~IEEE}\vspace{-2em}
		\thanks{ 
			This work is supported in part by MOE Singapore under Award T2EP50120-0024, National University of Singapore under Research Grant R-261-518-005-720, and The Guangdong Provincial Key Laboratory of Big Data Computing.
			(\it{Corresponding authors: Weidong Mei and Rui Zhang.})}
	\thanks{M. Fu is with the Department of Electrical and Computer Engineering, National University of Singapore, Singapore 117583 (e-mail:  fumin@nus.edu.sg).
    W. Mei is with the National Key Laboratory of Science and Technology on Communications, University of Electronic Science and Technologyof China, Chengdu 611731, China (e-mail: wmei@uestc.edu.cn) 
    R. Zhang is with School of Science and Engineering, Shenzhen Research Institute of Big Data, The Chinese University of Hong Kong, Shenzhen, Guangdong 518172, China (e-mail: rzhang@cuhk.edu.cn). He is also with the Department of Electrical and Computer Engineering, National University of Singapore, Singapore 117583 (e-mail: elezhang@nus.edu.sg).

}}
	
\maketitle

\setlength\abovedisplayskip{2pt}
\setlength\belowdisplayskip{2pt}
\setlength\abovedisplayshortskip{2pt}
\setlength\belowdisplayshortskip{2pt}
\setlength\arraycolsep{2pt}

\begin{abstract}
Intelligent reflecting surfaces (IRSs), active and/or passive, can be densely deployed in complex environments to significantly enhance wireless network coverage for both wireless information transfer (WIT) and wireless power transfer (WPT).
In this letter, we study the downlink WIT/WPT from a multi-antenna base station to a single-antenna user over a multi-active/passive IRS (AIRS/PIRS)-enabled wireless link.
In particular, we derive the location of the AIRS with those of the other PIRSs being fixed to maximize the received signal-to-noise ratio (SNR) and signal power at the user in the cases of WIT and WPT, respectively.
The derived solutions reveal that the optimal AIRS deployment is generally different for WIT versus WPT due to the different roles of AIRS-induced amplification noise.
Furthermore, both analytical and numerical results are provided to show the conditions under which the proposed AIRS deployment strategy yields superior performance to other baseline deployment strategies as well as the conventional all-PIRS enabled WIT/WPT.
\end{abstract}

\vspace{-1em}
\section{Introduction}
 
Intelligent reflecting surface (IRS) \cite{Wu2021Tutorial} has received high attention in wireless communications due to its passive, full-duplex and controllable signal reflection, which can improve the spectral and energy efficiency of future wireless networks cost-effectively.
Specifically, IRS consists of a large array of passive reflecting elements, each of which can be dynamically tuned to alter the phase/amplitude of its reflected signal.
By efficiently integrating IRSs into future wireless networks, it is anticipated that a quantum leap in system performance can be attained by jointly designing the operations of IRS and existing active nodes in wireless networks (e.g., base station (BS), user terminal, relay, and so on) \cite{Wu2021Tutorial}.

The appealing prospects of IRS have spurred extensive research to investigate its performance gain under different system setups (see e.g., \cite{Wu2021Tutorial, Mei2022Multireflection} and the references therein).
For example, in complex environments with scattered blockages, densely deploying IRSs can provide a pronounced path diversity gain for the BS to establish multiple blockage-free multi-IRS-reflection enabled line-of-sight (LoS) links with distributed users, especially when the direct links from the BS to them are severely blocked \cite{Mei2022Multireflection}.
This has thus motivated recent studies \cite{Mei2021Cooperative, Mei2022MIMO, Mei2021BeamTraining, Wang2022MultipleRIS, Liu2022intelligent} on the multi-IRS-reflection aided wireless networks (the use of passive IRS (PIRS) only), which are shown to achieve even higher cooperative passive beamforming (CPB) gains than conventional single-/double-PIRS aided systems.
Despite the high path diversity and CPB gains of multi-PIRSs, their signal coverage performance is severely limited by the multiplicative path loss of the cascaded channel due to the passive reflection of PIRSs \cite{Mei2022Multireflection}.

To tackle the above issue, in this letter, we consider the joint use of active IRS (AIRS) and PIRS for multi-IRS-enabled wireless links to enhance the performance of both wireless information transfer (WIT) and wireless power transfer (WPT).
Different from PIRS, AIRS is equipped with negative resistance components and hence able to reflect the incident signal with power amplification gain\cite{You2021active}, which compensates for the multiplicative path loss more effectively than CPB gain.
There have been a handful of recent works on AIRS/PIRS-aided communication systems.
For example, in \cite{Liang2022dual} and \cite{Fu2022active}, the authors studied a double-IRS-reflection system with a pair of AIRS and PIRS and optimized their cooperative beamforming and deployment.
In \cite{Zhang2022multi}, the authors studied a general multi-IRS-reflection system with one AIRS and multiple distributed PIRSs, to select the optimal reflection path from the BS to the user.
However, the above works either consider the deployment of AIRS in a simple two-IRS system \cite{Fu2022active} or assume a fixed location of AIRS \cite{Liang2022dual, Zhang2022multi} to optimize the WIT performance. 
	It thus remains unknown how to optimize the location of an AIRS in a general multi-IRS system. 
	In addition, they do not consider the optimization of  WPT performance, while the optimal AIRS deployment in WIT may not  apply to WPT due to their different objectives and roles played by AIRS-introduced amplification noise.

\begin{figure}[t]
	\centering
	\includegraphics[scale=0.3]{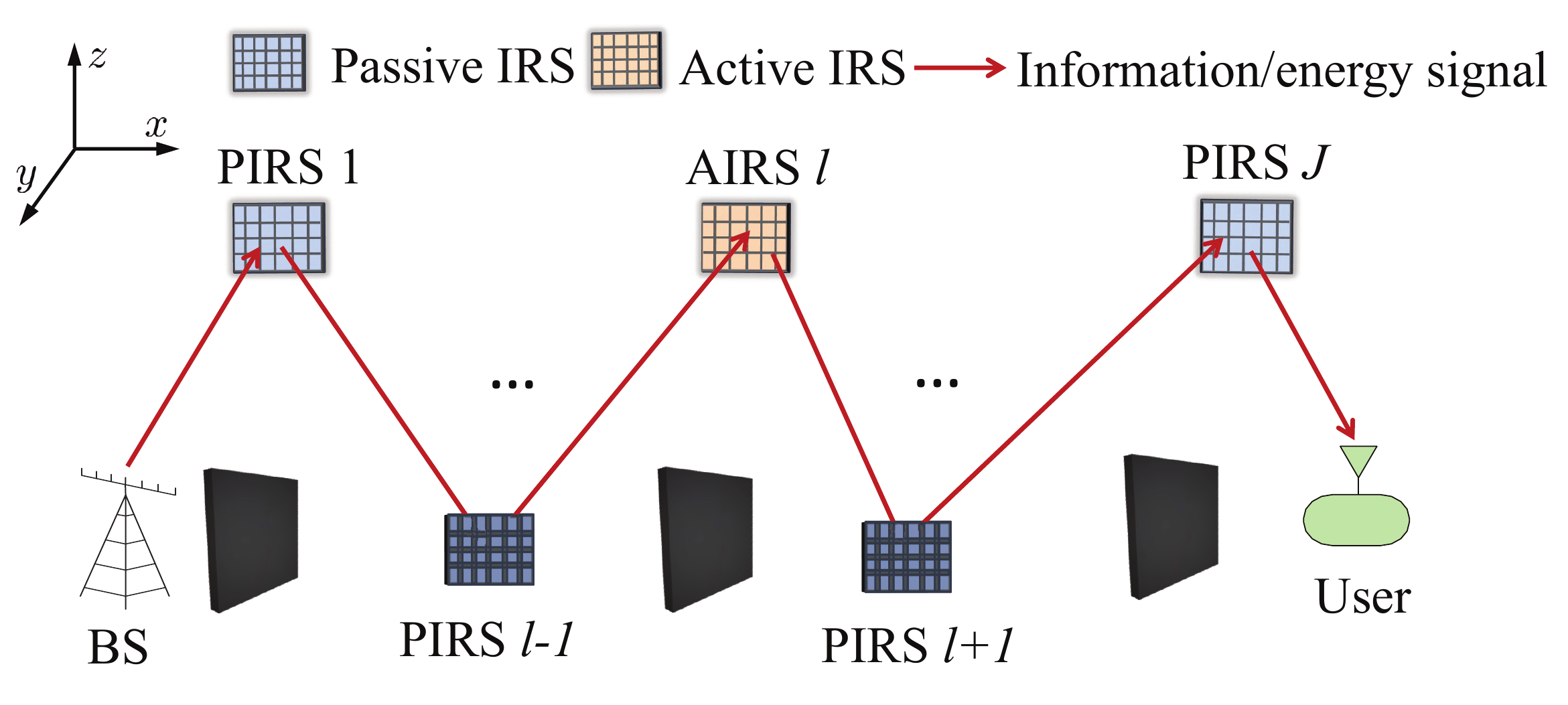}
	\vspace{-3mm}		
	\caption{Multi-active/passive-IRS enabled WIT/WPT.} \label{Fig:systemmodel}
	\vspace{-8mm}
\end{figure}

Inspired by the above, this letter studies the deployment of AIRS problem in multi-AIRS/PIRS enabled WIT/WPT, where a multi-antenna BS transmits information/energy wirelessly to a single-antenna user over a cascaded LoS link enabled by one AIRS and multiple PIRSs.
\textcolor{black}{As an initial study on this problem, we consider a simplified scenario with a single AIRS and a single user to facilitate the theoretical analysis of the AIRS deployment in WIT or WPT scenario.}
In particular, we aim to examine the deployment of the AIRS with those of the other PIRSs being fixed to maximize the received signal-to-noise ratio (SNR) or signal power at the user in the cases of WIT and WPT, respectively.
The optimal solutions in these two cases are derived in closed form, which reveals that the optimal AIRS deployment for WIT is generally different from that for WPT,  and depends on the number of reflecting elements per PIRS and the distance between two adjacent IRSs.
Furthermore, both analytical and numerical results are presented to show the conditions under which the proposed AIRS deployment solutions yield significant performance improvement than other baseline deployment strategies as well as the conventional all-PIRS enabled WIT/WPT.

 \vspace{-0.8em}     
 \section{System  Model}\label{Section:model}

 As shown in Fig.~\ref{Fig:systemmodel}, we consider the downlink WIT/WPT from an $M$-antenna BS to a single-antenna user with the help of $J$ IRSs, including $J-1$ PIRSs and one AIRS.
We consider a challenging scenario where, due to the dense and scattered obstacles in the environment, the BS can only transmit to the user over a multi-reflection path formed by these $J$ IRSs.
 For convenience, based on the distance from the BS to IRSs, we number them from 1 to $J$ (from the nearest to the farthest), with the AIRS numbered as $l$, $1 \le l \le J$, as depicted in Fig.~\ref{Fig:systemmodel}.
In this letter, we aim to study the optimal deployment of the AIRS over this multi-reflection path (i.e., the optimal index $l$) for achieving optimal WIT/WPT performance at the user.
We denote the set of $J$ IRSs and their formed multi-reflection path as ${\cal J}\triangleq \{1,2,\cdots, J\}$ and $\Omega$, respectively.
By properly placing these IRSs, we assume that an LoS link can be achieved between any two adjacent nodes in $\Omega$, such that a cooperative passive beamforming (CPB) gain can be reaped among them.
	The effects of all other randomly scattered links between the BS and IRSs are assumed to be negligible due to the lack of CPB and high multiplicative path loss over them \cite{Mei2022MIMO}.
As such, WIT/WPT performance depends only on the end-to-end channel gain of the LoS path $\Omega$.

Let $N_{p}/N_{a}$ denote the number of passive/active reflecting elements per PIRS/AIRS.
For each PIRS $k, k \in {\mathcal J}\!\setminus\! \{l\}$, let ${\bm \Phi}_k={\rm diag}\{e^{j\theta_{k,1}},\ldots,e^{j\theta_{k,N_{p}}}\} \in {\mathbb C}^{N_{p} \times N_{p}}$ denote its reflection coefficient matrix, where
$\theta_{k,n}\!\in\! [0,2\pi]$ denotes the phase shift of the $n$-th element, and its amplitude is set to one to maximize the reflected signal power \cite{Wu2021Tutorial}.
 While for the AIRS, since the incident signals to all of its reflecting elements experience the same path loss in the case of LoS channels, a common amplification factor $\eta$ can be used without performance loss, which also helps simplify the AIRS beamforming design.
Therefore, we denote its reflection coefficient matrix as $\eta {\bm \Phi}_l= {\rm diag}\{\eta e^{j\theta_{l,1}},\ldots,\eta e^{j\theta_{l, N_{a}}}\} \in {\mathbb C}^{N_{a} \times N_{a}}$.

Next, we characterize the LoS channel between any two adjacent nodes over $\Omega$.
For convenience, we refer to the BS and the user as nodes 0 and $J+1$, respectively.
Let $\bm H_{0,1}$, $\bm S_{k-1,k}$, and $\bm g^{\sf H}_{J, J+1}$ denote the baseband equivalent LoS channel from the BS to IRS 1, that from IRS $k-1$ to IRS $k$, $k \in \mathcal{J}$, and that from IRS $J$ to the user, respectively.
Let the BS-IRS 1 and IRS $J$-user distances be $d_{B}$ and $d_{U}$, respectively, while the distance between any two adjacent IRSs is assumed to be identical as $d_I$ (with a reference point selected at the BS and AIRS/PIRSs). Note that the results in this letter are also applicable to the case that the distances between any two adjacent IRSs are different.
It is also assumed that the above distances  are sufficiently large such that far-field propagation holds between any two adjacent nodes.
As such, the LoS channel between nodes $k$ and $k+1$  can be modeled as the product of transmit and receive array responses at the two sides.
Specifically, consider that the BS and each IRS are equipped with a uniform linear array (ULA) and a uniform planar array (UPA) parallel to the $x$-$z$ plane, respectively, as shown in Fig. \ref{Fig:systemmodel}.
For the ULA at the BS, its transmit array response with respect to (w.r.t.) IRS $1$ is written as $ {\tilde{\bm h}}_{0,1,t} \triangleq \bm u(\frac{2d}{\lambda}\cos\varphi^{t}_{0, 1}, M )\in \mathbb{C}^{M \times 1}$, where $\varphi^{t}_{0, 1}$ denotes the angle-of-departure (AoD) from the BS to IRS 1, $\lambda$ denotes the signal wavelength, $d$ denotes the spacing between two adjacent antennas/elements at the BS/each IRS, and $\bm u$ is the steering vector function defined as $\bm u (\varsigma, M')\triangleq [1,  e^{-j\pi\varsigma}, \ldots, e^{-j\pi(M' -1) \varsigma}]^{\sf T} \in \mathbb{C}^{M' \times 1}$.
For the UPA at each IRS, its transmit/receive array response vector can be expressed as the Kronecker product of two steering vector functions in the $x$- and $z$-axis directions, respectively.
In particular, let $N_{s,x}$ and $N_{s,z}$, $s\in \{p, a\}$, denote the numbers of elements at PIRS/AIRS in the $x$- and $z$-axis directions, respectively, with $N_{s} \!=\! N_{s,x} \!\times\! N_{s,z}$.
Then, the transmit array response of IRS $k$ w.r.t. node $k\!+\!1$ (IRS/user) is expressed as $ {\tilde{\bm s}}_{k,k+1,t}\triangleq \bm {u}(\frac{2d}{\lambda}\cos(\varphi^{t}_{k, k+1})\sin(\vartheta^t_{k, k+1}), N_{s,x}) \!\otimes\! \bm {u}(\frac{2d}{\lambda}\cos(\vartheta^t_{k, k+1}), N_{s,z})\in \mathbb{C}^{N_{s}\times 1}$, $s\in \{p,a\}$, where $\varphi^{t}_{k, k+1}$ and $\vartheta^t_{k, k+1}$ denote the azimuth and elevation AoDs from IRS $k$ to node $k+1$, respectively.
Similarly, we can define the receive array response of IRS $k+1$ w.r.t. node $k$ (IRS/BS) and denote it as ${\tilde{\bm s}}_{k,k+1,r}$.
As such, the BS-IRS 1, IRS $k$-IRS $k+1$, and IRS $J$-user LoS channels are respectively given by
\begingroup\makeatletter\def\f@size{9.1}\check@mathfonts\def\maketag@@@#1{\hbox{\m@th\normalsize\normalfont#1}}
\begin{eqnarray}
	\hspace{-1em}	{\bm H}_{0,1} &=& ({ \sqrt{\beta_0}}/{d^\frac{\alpha}{2}_{B}})e^{-\frac{j2\pi d_{B}}{\lambda}}{\tilde{\bm s}}_{0,1,r}{\tilde{\bm h}}^{\sf H}_{0,1,t}, \\
	\hspace{-1em}	{\bm S}_{k-1,k} &=& ({ \sqrt{\beta_0}}/{d^\frac{\alpha}{2}_{I}})e^{-\frac{j2\pi d_{I}}{\lambda}}{\tilde{\bm s}}_{k-1,k,r}{\tilde{\bm s}}^{\sf H}_{k-1,k,t}, k\!\in\! {\mathcal{J}}\!\setminus\!\{1\}, \\
	\hspace{-1em}	\bm g^{\sf H}_{J,J+1} &=& ({ \sqrt{\beta_0}}/{d_{U}^\frac{\alpha}{2}})e^{-\frac{j2\pi d_{U}}{\lambda}} {\tilde{\bm s}}^{\sf H}_{J,J+1,t},
\end{eqnarray}\normalsize
where $\alpha$ and $\beta_0$ respectively denote the LoS path-loss exponent and path gain at the reference distance of one meter (m).
We assume that these large-scale channel parameters have been measured through offline training \cite{Mei2021BeamTraining}.  Let $\bm w \in \mathbb{C}^{M\times 1}$ and $P_t$ denote the BS’s transmit beamforming vector and power, respectively, with $\|\bm w\|^2 = P_t$.
Then, given the index of the AIRS, $l$, its effective channels with the BS and the user can be respectively expressed as
\begingroup\makeatletter\def\f@size{9.1}\check@mathfonts\def\maketag@@@#1{\hbox{\m@th\normalsize\normalfont#1}}
\begin{eqnarray}
	\hspace{-1em} \bm h_{TA}(l) &=& \Big(\prod\nolimits_{k=1}^{l-1}\bm S_{k,k+1}\bm \Phi_{k}\Big)\bm H_{0, 1}\bm w \nonumber\\
	\hspace{-1em}	& = &\!\kappa_{B}\kappa^{l-1}_{I}e^{\frac{-j2\pi D_{1} }{\lambda}} {\tilde{\bm s}}_{l-1,l,r}\big({\tilde{\bm h}}^{\sf H}_{0,1,t}\bm w\big)\prod\nolimits_{k=1}^{l-1}A_{k}, \label{Eq:Tx-AIRS channel}\\
	\hspace{-1em}\bm h^{\sf H}_{AR}(l) &=& \bm g^{\sf H}_{J, J+1}\prod\nolimits_{k=l}^{J-1}\bm \Phi_{k+1}\bm S_{k, k+1} \nonumber\\
	\hspace{-1em}	& = & \kappa^{J-l}_{I}\kappa_{U}e^{\frac{-j2\pi D_{2} }{\lambda}} {\tilde{\bm s}}^{\sf H}_{J,J+1,t}\prod\nolimits_{k=l+1}^{J}A_{k}, \label{Eq:AIRS-Rx channel}
\end{eqnarray}\normalsize
where $A_{k} = {\tilde{\bm s}}^{\sf H}_{k,k+1,t}\bm \Phi_{k}{\tilde{\bm s}}_{k\!-\!1,k,r}, k \in \mathcal {J}$, $D_{1} \triangleq d_{B} + (l-1)d_{I}$ and $D_{2} \triangleq d_{U} + (J-l)d_{I}$ respectively denote BS-AIRS $l$ and AIRS $l$-user distances, and
$\kappa_{B} \triangleq { \sqrt{\beta_0}}/{d^\frac{\alpha}{2}_{B}}$, $\kappa_{I} \triangleq { \sqrt{\beta_0}}/{d^\frac{\alpha}{2}_{I}}$, and $\kappa_{U} \triangleq { \sqrt{\beta_0}}/{d^\frac{\alpha}{2}_{U}}$ denote the BS-IRS 1, inter-IRS, and IRS $J$-user LoS path gains, respectively.

Through the multi-reflection link $\Omega$, the received signal at the user is given by
\begin{eqnarray}\label{Eq:received signal}
	y_{r} = \bm h^{\sf H}_{AR}(l) \eta\bm \Phi_{l } (\bm h_{TA}(l) s + \bm{n}_{a}) + z,
\end{eqnarray}	
where $s$ denotes the transmitted information/energy symbol at the BS with $\mathbb{E}[\|s\|^2] = 1$,
$\bm{n}_{a}\in\mathbb{C}^{N_{a}\times 1}$ and $z\in \mathbb{C}$ denote the AIRS noise vector
and the AWGN noise, respectively.
It is assumed that $\bm{n}_{a}\sim \mathcal{CN}(\boldsymbol{0}_{N_{a}}, \sigma^2{\bf I}_{N_{a}})$ and $z\sim \mathcal{CN}(0,\sigma^2)$, where their power is assumed to be identical to $\sigma^2$.
Then, the reflected signal power by each AIRS element is given by $\eta^2\big(\kappa^2_B\kappa_I^{2(l-1)}|{\tilde{\bm h}}^{\sf H}_{0,1,t}\bm w|^2\prod_{k=1}^{l-1}|A_{k}|^2 + \sigma^2\big)$, which should satisfy the following power constraint:
\begin{eqnarray} \label{Cons:activeIRS}
	\eta^2\big(\kappa^2_B\kappa_I^{2(l-1)}|{\tilde{\bm h}}^{\sf H}_{0,1,t}\bm w|^2\prod\nolimits_{k=1}^{l-1}|A_{k}|^2 + \sigma^2\big) \leq P_{a},
\end{eqnarray} 
 where $P_a$ denotes the maximum amplification power per element of the AIRS.

 \vspace{-1.5em} 
 \section{Multi-AIRS/PIRS Enabled WIT}\label{Section:WIT}

 In this section, we first focus on the WIT case and derive the optimal BS/IRS beamforming in $\Omega$, as well as the optimal AIRS index $l$. Then, we analytically compare the optimal WIT performance with that of other baseline schemes to obtain important insights.

  \vspace{-1em}
 \subsection{Optimal BS/IRS Beamforming with given AIRS Deployment}\label{Section:WIT-beamforming}
 In the WIT case, the user aims to decode the information from its received signal in \eqref{Eq:received signal}.
 Hence, we should maximize its received SNR which is given by
 	\begingroup\makeatletter\def\f@size{9.2}\check@mathfonts\def\maketag@@@#1{\hbox{\m@th\normalsize\normalfont#1}}
 \begin{eqnarray}\label{Eq:SNR}
 	\hspace{-3em}	{\gamma}(l, \eta, \{\bm \Phi_k\},\bm w) = \frac{\eta^2\|\bm h^{\sf H}_{AR}(l) \bm \Phi_{l}\bm h_{TA}(l) \|^2}{\eta^2\|\bm h^{\sf H}_{AR}(l)\bm \Phi_{l}\|^2\sigma^2 +\sigma^2}.
 \end{eqnarray}\normalsize

 It is noted that the received SNR in \eqref{Eq:SNR} monotonically increases with $\eta^2$. Hence, for any given AIRS index $l$, the amplification power constraint in \eqref{Cons:activeIRS} must hold with equality to maximize \eqref{Eq:SNR}, which yields
 	\begingroup\makeatletter\def\f@size{9.2}\check@mathfonts\def\maketag@@@#1{\hbox{\m@th\normalsize\normalfont#1}}
 \begin{eqnarray}\label{Eq:AIRS amplitude}
 	\eta^2 &=& \frac{P_{a}}{\kappa^2_B\kappa_I^{2(l-1)}|{\tilde{\bm h}}^{\sf H}_{0,1,t}\bm w|^2\prod_{k=1}^{l-1}|A_{k}|^2 + \sigma^2}. 
 \end{eqnarray}\normalsize
 Furthermore, it can be verified that \eqref{Eq:SNR} monotonically increases with each $\lvert A_k\rvert, k \in \mathcal {J} $ and $|\tilde{\bm h}^{\sf H}_{0, 1, t} \bm w|$.
 As such, the optimal BS/IRS beamforming, i.e., $\bm w$ and $ \{\bm \Phi_k\}$, should be designed to maximize $|\tilde{\bm h}^{\sf H}_{0, 1, t} \bm w|$ and $\lvert A_k \rvert$, respectively, which results in
 \begin{eqnarray}
 	&& \bm w = \sqrt{P_t}{\tilde{\bm h}_{0, 1, t}}/{\|\tilde{\bm h}_{0, 1, t}\|}, \label{Eq:BS beam} \\
 	&&	\theta_{k,n} = \arg([\tilde{\bm s}_{k-1, k, r}]_n[\tilde{\bm s}^{\sf H}_{k, k+1, t}]_n), k \in \mathcal J, \forall n. \label{Eq:IRS phase}
 \end{eqnarray}
 By substituting \eqref{Eq:AIRS amplitude}-\eqref{Eq:IRS phase} into \eqref{Eq:SNR}, the user's received SNR can be simplified as
 \begin{eqnarray}\label{Eq:SNR-rewrite}
 	\hspace{-0.5em}	{\gamma}(l)
 	= \frac{C_aC_tN_a(N_{p}\kappa_I)^{2(J-1)}}{\sigma^2C_a(N_{p}\kappa_I)^{2(J-l)} +\sigma^2(C_t(N_{p}\kappa_I)^{2(l-1)} + \sigma^2)},
 \end{eqnarray}
 where $C_a \!=\! P_{a}N_a\kappa^2_{U}$, $C_t \!=\! P_{t}M\kappa^2_{B}$.
  In \eqref{Eq:SNR-rewrite}, it is noted that ${\gamma}(l) $ increases with $N_p$. In particular, as $N_p\rightarrow\infty$, we have
  \begingroup\makeatletter\def\f@size{9.2}\check@mathfonts\def\maketag@@@#1{\hbox{\m@th\normalsize\normalfont#1}}
 \begin{eqnarray}\label{Eq:SNR-Np}
 	{\gamma}(l)\rightarrow \begin{cases}
 		\frac{C_tN_a(N_{p}\kappa_I)^{2(l-1)}}{\sigma^2} = \mathcal{O}(N_{p}^{2(l-1)}) & \! {\text{if}~} l< \frac{J+1}{2},\\
 		\frac{C_aN_a(N_{p}\kappa_I)^{2(J-l)}}{\sigma^2} = \mathcal{O}(N_{p}^{2(J-l)})	&\! {\text{if}~} l\geq \frac{J+1}{2},
 	\end{cases}
 \end{eqnarray}\normalsize
 which shows that the CPB gain by the PIRSs over $\Omega$ is a piece-wise function of $l$.
 This is because due to the maximum per-element reflected power $P_a$, the information signal power and AIRS noise power at the user receiver scale with $N_p$ in the orders of $\mathcal{O}(N_{p}^{2(J-l)})$ and $\mathcal{O}(N_{p}^{2(J+1-2l)})$, respectively.
 Then, if $J+1-2l>0$, i.e., $l< \frac{J+1}{2}$, the amplification noise will dominate over the AWGN noise (which is of zero order of $N_p$).
 Thus, ${\gamma}(l)$ scales in the order of $\mathcal{O}(N_{p}^{2(l-1)})$, corresponding to the first case of \eqref{Eq:SNR-Np}.
 On the other hand, if $J+1-2l\leq0$, i.e., $l\geq \frac{J+1}{2}$, the zero-order AWGN noise will dominate over the amplification noise.
 Thus, ${\gamma}(l)$ has the same scaling behavior as the received information signal power, leading to the second case of \eqref{Eq:SNR-Np}.
 
 \vspace{-1.2em}                  
 \subsection{Problem Formulation}
Next, we formulate the AIRS deployment optimization problem to maximize the user's received SNR $\gamma(l)$.
First, to facilitate deployment optimization, we define
\begin{eqnarray}\label{Eq:f def}
	f(l) \triangleq \kappa^2_{B}(N_{p}\kappa_I)^{2(l-1)}, l \in \mathcal{J},
\end{eqnarray}
which denotes the effective channel power gain from the BS to the reference element of the AIRS in the case of a single-antenna BS, i.e., $M=1$. Note that due to the passive reflection of PIRSs, $f(l)$ monotonically decreases with $l$ (i.e., $N_p<\kappa^{-1}_I$) due to the increasing number of passive signal reflections and more power attenuation.
Then, the received SNR maximization problem can be formulated as
\begin{eqnarray}\label{problem:WIT-reduced}
	\hspace{-1.7em}\text{(P1)~}\mathop{ \text{maximize}}\limits_{l \in \cal J}	\frac{C_aC_tN_a(N_{p}\kappa_I)^{2(J-1)}}{\sigma^2\frac{C_a\kappa^2_{B}(N_{p}\kappa_I)^{2(J-1)}}{f(l)} +\sigma^2\frac{C_t}{\kappa^2_{B}} f(l) + \sigma^4}.
\end{eqnarray}
Note that as the numerator of \eqref{problem:WIT-reduced} is a constant, it is equivalent to minimizing its denominator by optimizing $l$.
However, as $l$ increases/decreases, $f(l)$ would decrease/increase, thus increasing/decreasing the first term in the denominator of \eqref{problem:WIT-reduced} while decreasing/increasing its second term.
It follows that the AIRS deployment $l$ should optimally balance these two terms to minimize the denominator of \eqref{problem:WIT-reduced}.
Next, the optimal solution to (P1) is derived in closed-form.
 
 \vspace{-1.3em}
 \subsection{Optimal Solution to (P1)}
 For (P1), we present its optimal solution in the following proposition.
 \begin{proposition}\label{Proposition: WIT}
 	The optimal solution to (P1), denoted as $l^\star_{1}$, is given by 
 	\begin{enumerate} 
 		\item Case I: if   ${C_a}<{C_t}$,
 		\begin{eqnarray}\label{Eq:solution-WIT-case1}
 			\hspace{-2em}	l^\star_{1} \!=\! 
 			\begin{cases}
 				\arg\max\limits_{l \in \{\lfloor \tilde{l} \rfloor , \lceil \tilde{l} \rceil \}}\gamma(l) &   {\text{if}~} N_p <  (\frac{C_a}{C_t})^\frac{1}{2(J-1)}\kappa^{-1}_I, \\
 				J	& {\text{if}~}(\frac{C_a}{C_t})^\frac{1}{2(J-1)}\kappa^{-1}_I \leq  N_p <\kappa^{-1}_I;\\ 
 			\end{cases}
 		\end{eqnarray}
 		\item Case II: if ${C_a}={C_t}$,  $l^\star_{1} = \arg\max\limits_{l \in \{\lfloor \frac{J+1}{2} \rfloor , \lceil \frac{J+1}{2} \rceil \}}\gamma(l)$;
 		\item Case III: if ${C_a}>{C_t}$,
 		\begin{eqnarray}\label{Eq:solution-WIT-case3}
 		\hspace{-2em}	l^\star_{1} \!=\! \begin{cases}
 				\arg\max\limits_{l \in \{\lfloor \tilde{l} \rfloor , \lceil \tilde{l} \rceil \}}\gamma(l)&   {\text{if}~}  N_p < (\frac{C_t}{C_a})^\frac{1}{2(J-1)}\kappa^{-1}_I, \\
 				1	& {\text{if}~}(\frac{C_t}{C_a})^\frac{1}{2(J-1)}\kappa^{-1}_I\leq N_p <\kappa^{-1}_I,\\
 			\end{cases}
 		\end{eqnarray}
 	\end{enumerate}
 	where $\tilde{l} = \frac{J+1}{2} + \frac{\log ({C_a/C_t})}{ 2\log(N^2_p\kappa^{2}_I)}$. 	  
 \end{proposition} 
 \begin{proof}
 	As previously mentioned, maximizing \eqref{problem:WIT-reduced} is equivalent to minimizing its denominator, i.e.,
 	\begin{eqnarray} \label{problem1:simplification}
 		\mathop{\text{minimize}}\limits_{l\in \mathcal{J}} &&  G(l) \triangleq \frac{C_a\kappa^2_{B}(N_{p}\kappa_I)^{2(J-1)}}{f(l)} +\frac{C_t}{\kappa^2_{B}} f(l), 
 	\end{eqnarray}
 	where the common scalar $\sigma^2$ and the constant $\sigma^4$ are omitted.
 	Next, we relax $l\in\mathcal{J}$ into $l \in [1,J]$ and take the derivative of $G(l)$ as
 	$	\frac{\partial G(l)}{\partial l} =  \frac{ C_tf^2(l)-   C_a\kappa^4_{B}(N_{p}\kappa_I)^{2(J-1)}}{\kappa^2_{B}f^2(l)} \frac{\partial f(l) }{\partial l}.$
 	Then, by letting $\frac{\partial G(l)}{\partial l} = 0$, we have $	\tilde{l} = \frac{J+1}{2} + \frac{\log (C_a/C_t)}{2\log(N^2_p\kappa^2_{I})}$.

 	As $f(l)$ decreases with $l$, we have $\frac{\partial f(l) }{\partial l}< 0$.
 	Thus, $G(l)$ decreases with $l$ when $l\in(0, \tilde{l} ]$ and increases with $l$ when $l\in [\tilde{l} ,+\infty]$.
 	However, $\tilde{l}$ is generally not an integer. Next, we elaborate on how to reconstruct the optimal integer $l$ for \eqref{problem1:simplification} based on $\tilde{l}$.

 	For Case I, if  ${C_a}<{C_t}$ and $N_p\geq (\frac{C_a}{C_t})^\frac{1}{2(J-1)}\kappa^{-1}_I $, we have $\tilde{l} \ge J$. As such, the optimal integer $l$ is given by $l^\star_{1} = J$, as shown in the second case of \eqref{Eq:solution-WIT-case1}.
 	On the other hand, if ${C_a}<{C_t}$ but  $N_p < (\frac{C_a}{C_t})^\frac{1}{2(J-1)}\kappa^{-1}_I$, we have $\frac{J+1}{2} \leq \tilde{l} < J$.
 	Thus, the optimal integer $l^\star_{1}$ should be either $\lfloor \tilde{l} \rfloor$ or $\lceil \tilde{l} \rceil$, which can be determined by comparing $\gamma(\lfloor \tilde{l} \rfloor)$  with $\gamma(\lceil \tilde{l} \rceil)$ and thus
 	gives rise to the first case of \eqref{Eq:solution-WIT-case1}.
 	
 	Cases II and III can be proved similarly to Case I and the details are omitted for brevity.
 \end{proof}
 \vspace{-0.5em} 
 Note that Case I (i.e, ${C_a}<{C_t}$ or $N_a < {P_tM\kappa_{B}^2}/{(P_a\kappa_{U}^2)}$) usually holds in practice since ${P_tM\kappa_{B}^2}/{(P_a\kappa_{U}^2)}$ is generally large due to $MP_t\gg P_a$, while $N_a$ is generally small or moderate to reduce the hardware cost of AIRS.
 It is also worth noting that in Case I, as $\tilde {l} > (J+1)/2$,
 the AIRS should be deployed in the second half of the IRSs.
 Furthermore, it is observed from \eqref{Eq:solution-WIT-case1} that 
 with increasing/decreasing $N_p\kappa_I$, the AIRS should be deployed closer to the final/middle IRS.
  Similarly, it is observed from \eqref{Eq:solution-WIT-case3} that in Case III, the AIRS should be deployed in the first half of the IRSs and closer to the BS/middle IRS with increasing/decreasing $N_p\kappa_I$.

 \vspace{-1em} 
 \subsection{Performance Comparison} \label{subsection:WIT}
  In the following, we analytically compare the optimal WIT performance with that by the following two benchmark schemes. 
 The first scheme deploys the AIRS at the middle of all IRSs, i.e., $l = \frac{J+1}{2}$ (assuming an odd $J$ for convenience), while the second one replaces the AIRS with a PIRS. 
 Note that for fairness of comparison, we set the BS's transmit power as $P_t+N_aP_a$ in the second benchmark scheme.
 By substituting $l^\star_{1}$ into \eqref{Eq:SNR-rewrite}, the maximum user's received SNR, denoted as ${\gamma}^\star_{\text{A}}$,  is given by \begin{eqnarray}\label{Eq:SNR-optimal}
 	\hspace{-1em}	{\gamma}^\star_{\text{A}} 	 
 	=\frac{C_aC_tN_a(N_{p}\kappa_I)^{2(J-1)}}{\sigma^2C_a(N_{p}\kappa_I)^{2(J-l^\star_{1})} +\sigma^2C_t(N_{p}\kappa_I)^{2(l^\star_{1}-1)} + \sigma^4}.
 \end{eqnarray}

 \subsubsection{Comparison with Scheme 1}
 With $l = \frac{J+1}{2}$, the user's received SNR in Scheme 1 is given by \begin{eqnarray}\label{Eq:SNR-mid}
 	{\gamma}^{\text{mid}}_{\text{A}}	 
 	=\frac{C_aC_tN_a(N_{p}\kappa_I)^{2(J-1)}}{\sigma^2C_a(N_{p}\kappa_I)^{(J-1)} +\sigma^2C_t(N_{p}\kappa_I)^{(J-1)} + \sigma^4}.
 \end{eqnarray}
 By comparing ${\gamma}^\star_{\text{A}}$ in \eqref{Eq:SNR-optimal} with ${\gamma}^{\text{mid}}_{\text{A}}$ in \eqref{Eq:SNR-mid}, we have
 \begingroup\makeatletter\def\f@size{9.2}\check@mathfonts\def\maketag@@@#1{\hbox{\m@th\normalsize\normalfont#1}}
 \begin{eqnarray}
 	\frac{{\gamma}^\star_{\text{A}}}{{\gamma}^{\text{mid}}_{\text{A}}}
 	&\!=\!&\frac{C_a(N_{p}\kappa_I)^{(J-1)} + C_t(N_{p}\kappa_I)^{(J-1)} + \sigma^2}{C_a(N_{p}\kappa_I)^{2(J-l^\star_{1})} +C_t(N_{p}\kappa_I)^{2(l^\star_{1}-1)} + \sigma^2} \nonumber \\
 	&{\geq}\! & \frac{C_a + C_t+ \sigma^2(N_{p}\kappa_I)^{(1-J)}}{C_a(N_{p}\kappa_I)^{(1-J)}\!+\!C_t(N_{p}\kappa_I)^{(J-1)}\!+\!\sigma^2(N_{p}\kappa_I)^{(1-J)}}\nonumber\\
 	& \overset{\sigma^2 \rightarrow 0}{\longrightarrow}& (N_{p}\kappa_I)^{(1-J)} \Big(1+ \frac{C_a-C_a(N_{p}\kappa_I)^{2(1-J)}}{C_a(N_{p}\kappa_I)^{2(1-J)}+C_t}\Big),
 	\label{NEq:SNR-optimal-mid}
 \end{eqnarray} \normalsize
 where the inequality is due to ${\gamma}^\star_{\text{A}}\geq {\gamma}(J)$.
 It is noted that the right-hand side (RHS) of \eqref{NEq:SNR-optimal-mid} increases with ${C_t}$ and decreases with ${C_a}$, which implies that in the high transmit power regime, the performance gain of the optimal deployment over Scheme 1 may become more significant as $M$ or $P_t$ increases and/or $N_a$ or $P_a$ decreases.

 \subsubsection{Comparison with Scheme 2}
 In Scheme 2 with all $J$ PIRSs, the maximum received SNR is given by \cite{Mei2021Cooperative}
 \begin{eqnarray}\label{Eq:all-passive-SNR}
 	\gamma_{\text{P}} = {(P_{\rm t} + N_aP_{a})M\kappa^2_B\kappa^2_UN^{2J}_{p}\kappa^{2(J-1)}_{I}/\sigma^2}. 
 \end{eqnarray}
 By comparing \eqref{Eq:all-passive-SNR} with \eqref{Eq:SNR-optimal},  it holds that $\frac{\gamma^\star_{\text{A}}}{\gamma_{\text{P}}}$ equals to
  \begingroup\makeatletter\def\f@size{9.2}\check@mathfonts\def\maketag@@@#1{\hbox{\m@th\normalsize\normalfont#1}}
 \begin{eqnarray}\label{Neq:WIT comparison}
 \hspace{-2em}&&	\frac{N_a^2P_tP_a}{(P_t \!+\! N_aP_a)N_p^2(C_a(N_{p}\kappa_I)^{2(J\!-\!l^\star_{1})} \!+\!C_t(N_{p}\kappa_I)^{2(l^\star_{1}\!-\!1)} \!+\! \sigma^2)} \nonumber\\
 \hspace{-2em} &&\geq \frac{N_a^2P_tP_a}{(P_t + N_aP_a)N_p^2(C_a \!+\!C_t(N_{p}\kappa_I)^{2(J-1)} + \sigma^2)},
 \end{eqnarray}	\normalsize
where we have applied ${\gamma}^\star_{\text{A}}\geq {\gamma}(J)$ again.
 It is noted that the RHS of \eqref{Neq:WIT comparison} decreases with $N_p$, which implies that the multi-AIRS/PIRS enabled system outperforms the all-PIRS system if $N_p$ is small. 
 This is expected, as the CPB gain by PIRSs is low and insufficient to compensate for the severe multiplicative path loss in the all-PIRS system with small $N_p$. However, the AIRS's amplification gain in the multi-AIRS/PIRS system compensates for the path loss more effectively.

 \vspace{-0.8em}
  \section{Multi-AIRS/PIRS Enabled WPT} \label{Section:WPT}
 In this section, we focus on the BS/IRS beamforming and AIRS deployment designs in the WPT case.
 \vspace{-1em}
 \subsection{Optimal BS/IRS Beamforming with given AIRS Deployment}\label{Section:WPT-beamforming}
In the WPT case,  the total received power at the user receiver is given by (excluding the receiver noise)
\begingroup\makeatletter\def\f@size{9.5}\check@mathfonts\def\maketag@@@#1{\hbox{\m@th\normalsize\normalfont#1}}
 \begin{eqnarray}\label{Eq:power}
 	\!	\!	\!{Q}(l, \eta, \{\!\bm \Phi_k \!\},\bm w) \!=\!  \eta^2\!(|\bm h^{\sf H}_{AR}(l) \bm \Phi_{l } \bm h_{TA}(l) |^2\!+\! \sigma^2\|\bm h^{\sf H}_{AR}(l)\|^2). 
 \end{eqnarray}\normalsize
Observed from \eqref{Eq:power}, the user can harvest power from both the  energy signal and the AIRS amplification noise.

Similar to Section \ref{Section:WIT}, it can be shown that the optimal AIRS amplification factor and BS/IRS beamforming that maximize \eqref{Eq:power} should satisfy \eqref{Eq:AIRS amplitude}-\eqref{Eq:IRS phase}.
By plugging them into  \eqref{Eq:power}, the user's received power can be simplified as
 \begin{eqnarray}\label{Eq:power-rewrite}
 	{Q}(l) =
 	\frac{C_aC_tN_a(N_{p}\kappa_I)^{2(J-1)} + \sigma^2C_a(N_{p}\kappa_I)^{2(J-l)}}{ C_t(N_{p}\kappa_I)^{2(l-1)} + \sigma^2 }.
 \end{eqnarray}
 It follows from \eqref{Eq:power-rewrite} that as $N_p\rightarrow\infty$, we have
 \begin{eqnarray}\label{Eq:power-Np}
 	{ Q}(l)\rightarrow C_aN_a(N_{p}\kappa_I)^{2(J-l)} = \mathcal{O}(N_{p}^{2(J-l)}),
 \end{eqnarray}
 which implies that similarly as in the WIT case,  the CPB gain by the PIRSs over $\Omega$ depends on $l$. 
The reason is that the signal power and AIRS noise power at the user receiver scale with $N_p$ in the orders of $\mathcal{O}(N_{p}^{2(J-l)})$ and $\mathcal{O}(N_{p}^{2(J+1-2l)})$, respectively.
Since the order of the former is no smaller than that of the latter,  the total received power should have the same scaling behavior as the former, i.e., $\mathcal{O}(N_{p}^{2(J-l)})$.

 \vspace{-1em}
 \subsection{Problem Formulation}
 Based on the above,  the received power maximization problem by optimizing $l$ is formulated as
 \begin{eqnarray}\label{problem:case2-reduced}
 \hspace{-1em}\text{(P2)~}	\mathop{ \text{maximize}}\limits_{l \in \cal J} &&   \frac{C_a(N_{p}\kappa_I)^{2(J-1)}(C_tN_af(l) + \sigma^2\kappa^2_{B})}{\frac{C_t}{\kappa^2_{B}}f^2(l) + \sigma^2f(l)},
 \end{eqnarray}
where $f(l)$ is defined in \eqref{Eq:f def}.
 It is observed from \eqref{problem:case2-reduced} that as $l$ increases, both its numerator and denominator would decrease.
 Next, the optimal solution to (P2) is derived.
 
 \vspace{-1em}
 \subsection{Optimal Solution to (P2)}
 For (P2), we present its optimal solution  in the following proposition.
 \begin{proposition}\label{Proposition: WPT}
 	The optimal solution to (P2),  denoted as $l^\star_{2} $, is given by $l^\star_{2} = J$.
 \end{proposition}
 \begin{proof}
 	The derivative of $Q(l)$ w.r.t. $l$ can be calculated as $\frac{\partial }{\partial l}Q(l)\!=
 		\!	 -C_a(N_{p}\kappa_I)^{2(J-1)}\frac{{C^2_tN_a}f^2(l) + 2{C_t\sigma^2}f(l)+\sigma^4\kappa^2_{B}}{\kappa^2_{B}(\frac{C_t}{\kappa^2_{B}}f^2(l) + \sigma^2f(l))^2}\frac{\partial }{\partial l}f(l)$.
 	As $f(l)$ monotonically decreases with $l$, we have $\frac{\partial }{\partial l}f(l) < 0$, thus resulting in $\frac{\partial }{\partial l}Q(l)>0, \forall l$. Hence, ${Q}(l)$ increases with $l$, and the optimal solution to (P2) is given by $l=J$. 
 \end{proof}
 \vspace{-0.5em} 
 Interestingly, Proposition \ref{Proposition: WPT} implies that different from WIT, AIRS should always be deployed at the location of the final IRS in WPT to maximize the received power.
 This is expected, since by deploying the AIRS closest to the user, its incident signal power can be most severely attenuated, thereby achieving the highest AIRS amplification gain for any given $P_a$. Meanwhile, this least attenuates the AIRS amplification noise power for any given $\sigma^2$, which is in favor of WPT at the user as well. It is worth noting that in the conventional single-PIRS system, it has been shown in \cite{Wu2022WIT-WPT} that the optimal PIRS deployment may be different in the case of WIT versus WPT in the multi-user scenario. While our results show that the optimal AIRS deployment may be different in these two scenarios even in the single-user case.

\begin{figure*}[t]
	\centering
	\begin{minipage}{.32\textwidth}
		\centering		
		\includegraphics[scale=0.35]{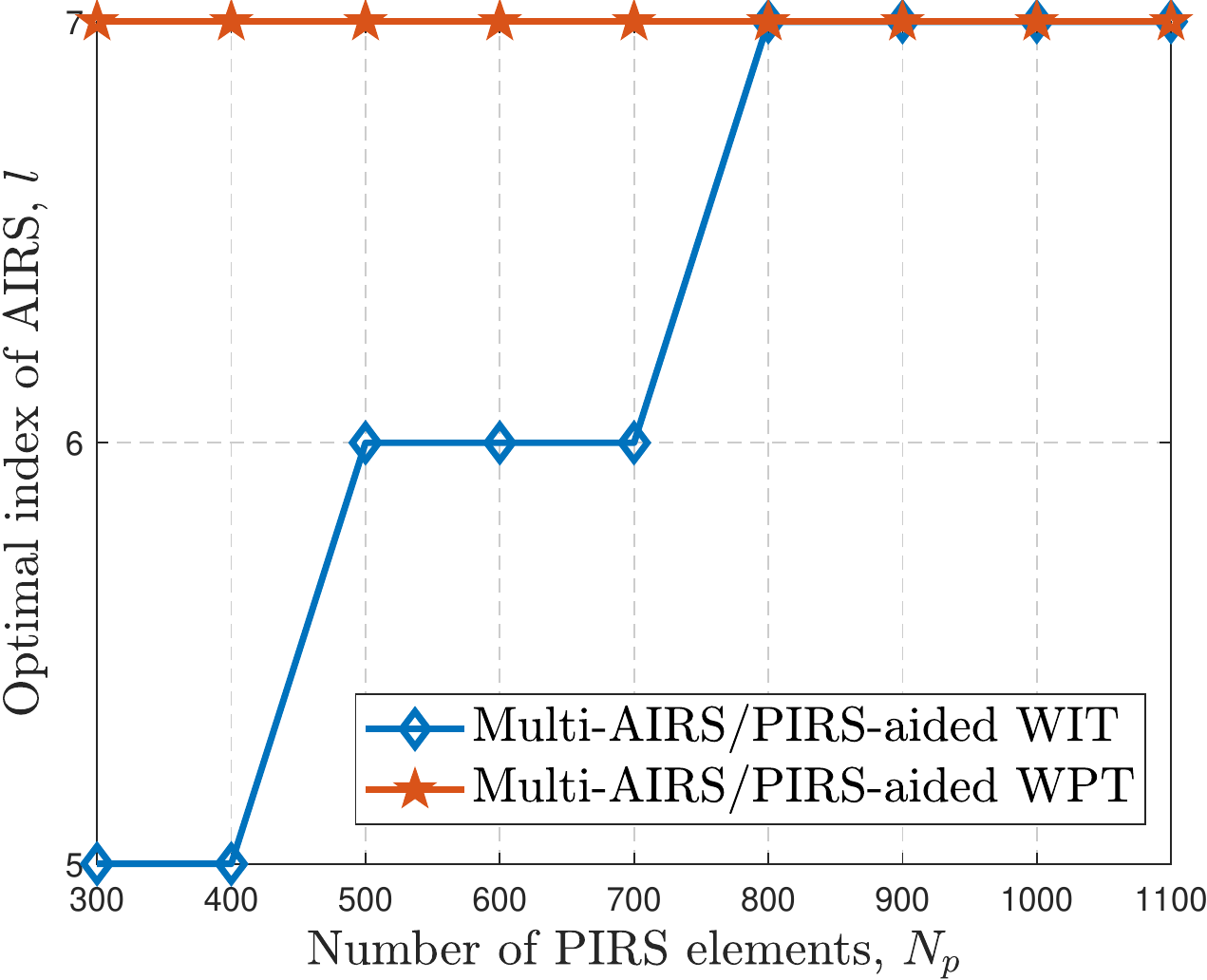}	
		\vspace{0mm}		
		\caption{Optimal AIRS deployment versus the number of elements per PIRS.} \label{Fig:index-Np}	
		\vspace{-5mm}
	\end{minipage}
	\hspace{1mm}
	\begin{minipage}{.32\textwidth}
		\centering
		\includegraphics[scale=0.35]{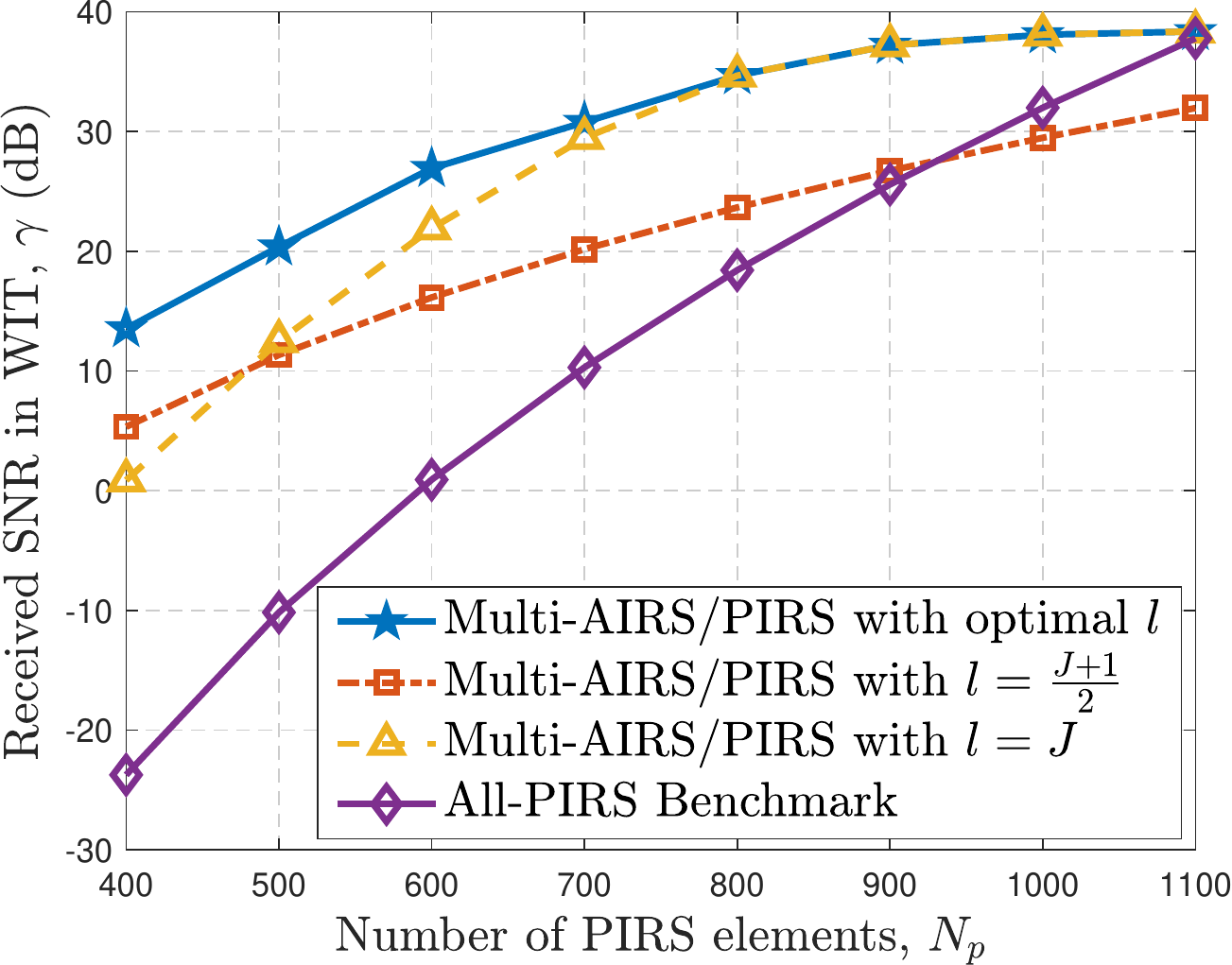}	
		\vspace{0mm}	
		\caption{Received SNR versus the number of elements per PIRS.} \label{Fig:WIT-Np}
		\vspace{-5mm}
	\end{minipage}
	\hspace{1mm}
	\begin{minipage}{.32\textwidth}
		\centering		
		\includegraphics[scale=0.35]{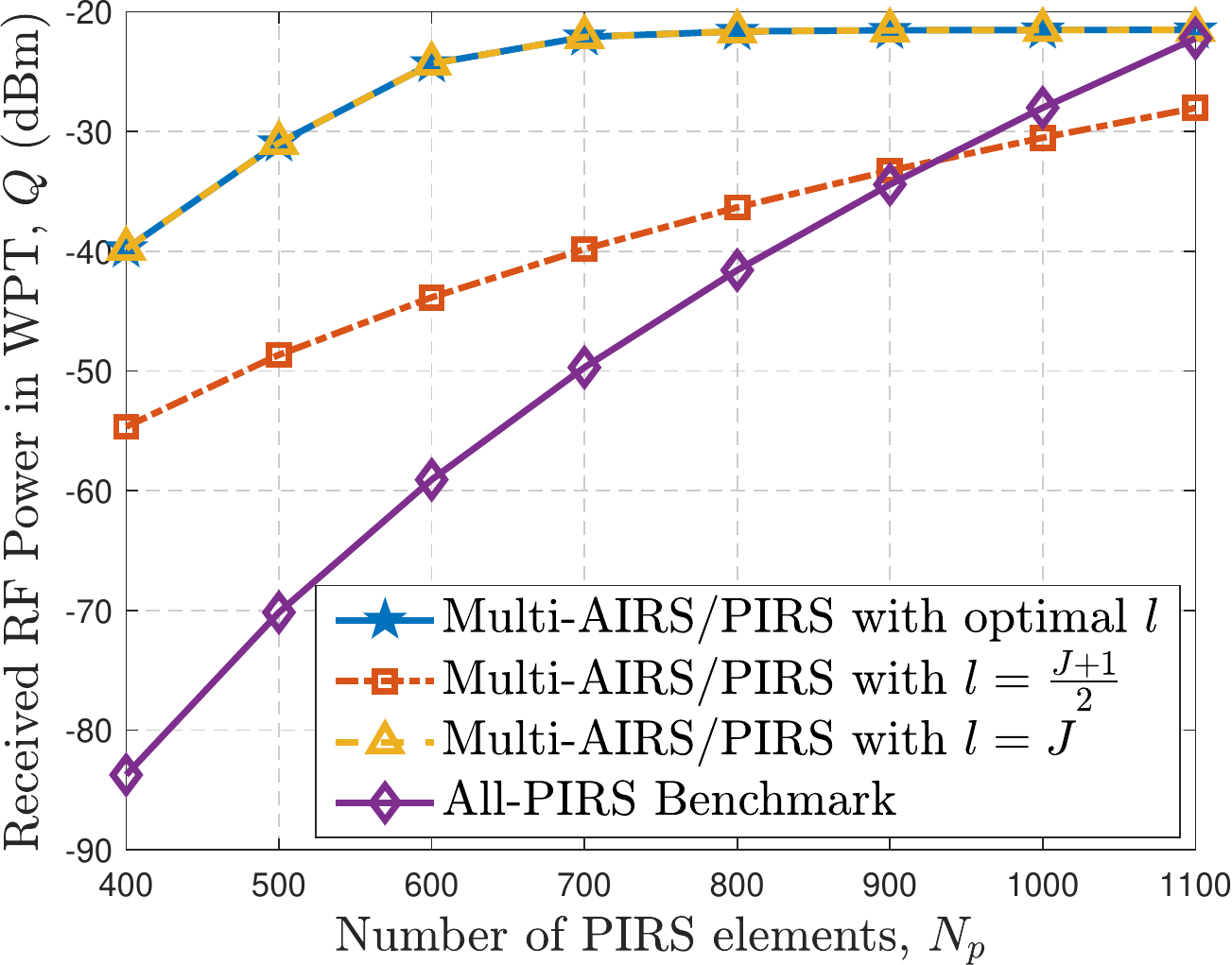}	
		\vspace{0mm}		
		\caption{Received power versus the number of elements per PIRS.} \label{Fig:WPT-Np}	
		\vspace{-5mm}
	\end{minipage}
	
\end{figure*} 
 
 \vspace{-1em}
 \subsection{Performance Comparison}
 
Similar to the WIT case, we next compare the optimal WPT performance with that by the two benchmark schemes presented in Section \ref{subsection:WIT}.  
By substituting $l^{\star}_{2} = J$ into \eqref{Eq:power-rewrite}, the maximum received power is given by \begin{eqnarray}\label{Eq:optimal power}
	Q_{\text{A}}^\star	 
	=\frac{C_a(C_tN_a(N_{p}\kappa_I)^{2(J-1)} + \sigma^2)}{ C_t(N_{p}\kappa_I)^{2(J-1)} + \sigma^2 }.
\end{eqnarray}

\subsubsection{Comparison with Scheme 1}
\begingroup
\allowdisplaybreaks
In Scheme 1, by substituting $l = \frac{J+1}{2}$ into \eqref{Eq:power-rewrite}, the received power is obtained as
\begin{eqnarray}\label{Eq:Q-mid}
{Q}_{\text{A}}^{\text{mid}} 	 
	=\frac{C_a(C_tN_a(N_{p}\kappa_I)^{2(J-1)} + \sigma^2(N_{p}\kappa_I)^{(J-1)})}{ C_t(N_{p}\kappa_I)^{(J-1)} + \sigma^2 }.
\end{eqnarray}
By comparing \eqref{Eq:optimal power} with  \eqref{Eq:Q-mid}, we have 
\begin{eqnarray}\label{NEq:power-optimal-mid}
\frac{Q_{\text{A}}^\star}{{Q}_{\text{A}}^{\text{mid}} }= (N_{p}\kappa_I)^{(1-J)}  \big(1+\rho_1\big) 
\overset{\sigma^2 \rightarrow 0}{\longrightarrow} (N_{p}\kappa_I)^{(1-J)},
\end{eqnarray} 
where $\rho_1\! =\! \frac{\sigma^2C_t(N_{p}\kappa_I)^{(J\!-\!1)}((1-N_{p}\kappa_I)^{(J\!-\!1)})(1-N_a) }{C^2_tN_a(N_{p}\kappa_I)^{3(J\!-\!1)} \!+\! C_t\sigma^2(N_{p}\kappa_I)^{(J\!-\!1)}(N_a \!+\! (N_{p}\kappa_I)^{(J\!-\!1)} )  \!+\! \sigma^4}$.
This implies that in the high transmit power regime, the performance gain of the former over the latter becomes more significant as $J$ increases and/or $N_p$ decreases.

  \subsubsection{Comparison with Scheme 2}
 In Scheme 2, the maximum received power is given by (excluding the receiver noise)
 \begin{eqnarray}\label{Eq:all-passive-power}
 	&&Q_{\text{P}} = (P_{\rm t} + N_aP_{a})M\kappa^2_B\kappa^2_UN^{2}_{p}(N_{p}\kappa_I)^{2(J-1)}.
 \end{eqnarray}	
 By comparing \eqref{Eq:optimal power} with  \eqref{Eq:all-passive-power}, we have
  \begingroup\makeatletter\def\f@size{9.2}\check@mathfonts\def\maketag@@@#1{\hbox{\m@th\normalsize\normalfont#1}}
\begin{eqnarray}
\frac{{Q}_{\text{A}}^\star}{Q_{\text{P}}} &=& \frac{C_aN_a}{N^{2J}_{p}\kappa_I^{2(J-1)}(C_t\kappa_U^2 + C_aM\kappa_B^2)}(1+\rho_2) \nonumber\\
&\overset{\sigma^2 \rightarrow 0}{\longrightarrow}&
	\frac{C_aN_a}{N^{2J}_{p}\kappa_I^{2(J-1)}(C_t\kappa_U^2 + C_aM\kappa_B^2)},\label{NEq:power-optimal-all}
\end{eqnarray} \normalsize
where $\rho_2 = \frac{\sigma^2(1-N_a)}{C_tN_a(N_{p}\kappa_I)^{2(J-1)} + N_a\sigma^2}$.
By letting  the RHS of \eqref{NEq:power-optimal-all} be larger than one, it follows that in the high transmit power regime, we have ${Q}_{\text{A}}^\star> Q_{\text{P}} $ if and only if
 \begin{eqnarray}\label{Neq:WPT comparison-Np}
	N_p < \Big(\frac{N_a^2 P_a}{P_{\rm t}M\kappa^2_B + P_{a}N_aM\kappa^2_B}\Big)^{\frac{1}{2J}}{\kappa_{I}^{\frac{1-J}{J}}},
\end{eqnarray}	
 which implies that the multi-AIRS/PIRS system outperforms the all-PIRS system if $N_p$ is small in WPT, which is the same as the WIT case.

  \vspace{-1em}
\section{Numerical Results}\label{Section:results}
In this section, we present numerical results to validate the theoretical analysis in this letter.
Unless specified otherwise, the simulation parameters are set as follows.
The number of IRSs is set as $J = 7$ with $d_{B} = 4$ m, $d_{U} = 4$ m, and $d_{I} = 10$ m.
The carrier frequency is set to 3.5 GHz and thus, the reference LoS path gain is $\beta_0 = (\lambda/4\pi)^2 = -43$ dB, with the wavelength $\lambda = 0.087$ m.
Other simulation parameters include $P_{\rm t} = 30$ dBm, $P_{a} = -10$ dBm, $\sigma^2 = -60$ dBm,  $\alpha = 2$,  and $M = 10$.
The number of AIRS reflecting elements is set as $N_{a} = 150$, under which Case I in \eqref{Eq:solution-WIT-case1} holds.

Fig. \ref{Fig:index-Np} shows the optimal index of the AIRS in WIT/WPT versus the number of  elements per PIRS. 
It is observed that in the WIT case, the optimal AIRS index $l$ is closer to $J$ as $N_p$ increases; while in the WPT case, it  always equals $J$. 
The above results are consistent with our theoretical results in Propositions \ref{Proposition: WIT} and \ref{Proposition: WPT}, respectively.

Fig.\,\ref{Fig:WIT-Np} shows the received SNRs in WIT by different schemes versus $N_p$.
First, it is observed that the performance gap between optimal $l$ (i.e., $l=l_1^\star$) and $l=J$ shrinks with $N_p$, which is consistent with our analytical result in \eqref{Eq:solution-WIT-case1}.
Moreover, the optimal AIRS deployment always yields better performance than deploying it at the middle IRS (i.e., $l\!=\!\frac{J+1}{2}$).
Furthermore, all of the above schemes are observed to significantly outperform the all-PIRS system when $N_p$ is sufficiently small, which is consistent with our result in Section \ref{subsection:WIT}.

Fig.\,\ref{Fig:WPT-Np} shows the received power in WPT by different schemes versus $N_p$.
As expected, it is observed that the optimal AIRS deployment yields the same performance as that with $l=J$ and significantly outperforms that with $l\!=\!\frac{J+1}{2}$. Moreover, similar to the observations made in Fig. \ref{Fig:WIT-Np}, the multi-AIRS/PIRS system can yield a high performance gain over the all-PIRS system when $N_p$ is small.
The observations from Figs. 3-4 indicate that the hybrid multi-IRS-reflection with optimized AIRS deployment can yield better performance than the all-PIRS reflection, and the AIRS needs to be deployed at different locations in the case of WIT versus WPT.

  \vspace{-0.5em}
\section{Conclusion}\label{Section:conclusion}
In this letter, we study the optimal AIRS deployment for multi-AIRS/PIRS enabled  WIT and WPT. 
It is shown that the optimal AIRS deployment for WPT is always at the final IRS, while this is not the case for WIT in general.
We also derive the performance for WIT/WPT analytically and demonstrate the importance of AIRS deployment in the multi-AIRS/PIRS system as well as its performance gain over the conventional all-PIRS system.			
It is interesting to study the AIRS deployment design in the cases with simultaneous wireless information and power transfer (SWIPT), multiple AIRSs, and/or multiple distributed users in future work.

 \vspace{-0.5em}

\bibliographystyle{IEEEtran}
\bibliography{ref} 

\end{document}